\documentclass[11pt,letter]{article}
\usepackage{authblk}
\usepackage[margin=1in]{geometry}
\usepackage[utf8x]{inputenc}
\usepackage{amsmath, amsthm}
\usepackage{wrapfig,floatflt,graphicx,amssymb,textcomp,array,amsmath}
\usepackage{algpseudocode} 
\usepackage{enumerate}
\usepackage{multirow}
\usepackage{tabularx}
\usepackage{algorithm}
\usepackage{color}
\usepackage{todonotes}
\usepackage[titletoc,title]{appendix}

\newcommand{\etal}{{et~al.}}
\newcommand{\DT}[1]{DT_1(#1)}
\newcommand{\DelT}[1]{DT(#1)}
\newcommand{\UDG}[1]{U\!DG(#1)}
\newcommand{\ray}[1]{\overrightarrow{#1}}
\newcommand{\pth}[2]{$(#1,#2)$-path}
\newcommand{\walk}[2]{$(#1,#2)$-walk}
\newcommand{\ratio}{341}
\title{Plane Hop Spanners for Unit Disk Graphs\\ \Large Simpler and Better
}

\author{Ahmad Biniaz\thanks{Supported by NSERC Postdoctoral Fellowship}
}

\affil{Cheriton School of Computer Science\\University of Waterloo\\ \texttt{ahmad.biniaz@gmail.com}}

\date{}
\newtheorem{lemma}{Lemma}
\newtheorem{corollary}{Corollary}

\newtheorem{theorem}{Theorem}

\newtheorem*{problem*}{Problem}
\newtheorem*{claim*}{Claim}
\newtheorem*{invariant*}{Invariant}

\begin{document}
	\maketitle
	\begin{abstract}
	The unit disk graph (UDG) is a widely employed model for the study of wireless networks. In this model, wireless nodes are represented by points in the plane and there is an edge between two points if and only if their Euclidean distance is at most one. A {\em hop spanner} for the UDG is a spanning subgraph $H$ such that for every edge $(p,q)$ in the UDG the topological shortest path between $p$ and $q$ in $H$ has a constant number of edges. The {\em hop stretch factor} of $H$ is the maximum number of edges of these paths. A hop spanner is {\em plane} (i.e. embedded planar) if its edges do not cross each other.
	
	The problem of constructing hop spanners for the UDG has received considerable attention in both computational geometry and wireless ad hoc networks. Despite this attention, there has
	not been significant progress on getting hop spanners that (i) are plane, and (ii) have low hop stretch factor. Previous constructions either do not ensure the planarity or have high hop stretch factor. The only construction that satisfies both conditions is due to Catusse, Chepoi, and Vax{\`{e}}s (2010); their plane hop spanner has hop stretch factor at most 449. 
	
	Our main result is a simple algorithm that constructs a plane hop spanner for the UDG. In addition
	to the simplicity, the hop stretch factor of the constructed spanner is at most $\ratio$. Even though the algorithm itself is simple, its analysis is rather involved. Several results on the plane geometry are established in the course of the proof. These results are of independent interest.  
	\end{abstract}
\section{Introduction}
Computational geometry techniques are widely used to solve problems, such as topology construction, routing, and broadcasting, in wireless ad hoc networks. A wireless ad hoc network is usually modeled as a unit disk graph (UDG). In this model wireless devices are represented by points in the plane and assumed to have identical unit transmission radii. There exists an edge between two points if their Euclidean distance is at most one unit; this edge indicates that the corresponding devices are in each other's transmission range and can communicate.

A {\em geometric graph} is a graph whose vertices are points in the plane and whose edges are
straight-line segments between the points. 
A geometric graph is {\em plane} if its edges do not cross each other. Let $G$ be a geometric graph. 
A {\em topological shortest path} between any two vertices $u$ and $v$ in $G$ is a path that connects $u$ and $v$ and has the minimum number of edges. The {\em hop distance} $h_G(u,v)$ between $u$ and $v$ is the number of edges of a topological shortest path between them. 

For a point set $P$ in the plane, the {\em unit disk graph} $\UDG{P}$ is a geometric graph with vertex set $P$ that has an edge between two points $p$ and $q$ if and only if their Euclidean distance $|pq|$ is at most $1$. 
A {\em hop spanner} for $\UDG{P}$ is a spanning subgraph $H$ such that for any edge $(p,q)\in\UDG{P}$ it holds that $h_H(p,q)\leqslant t$, where $t$ is some positive constant. The constant $t$ is called the {\em hop stretch factor} of $H$. In this paper we study the problem of constructing UDG hop spanners
that are plane and have low hop stretch factor.  

The {\em Euclidean spanner} and {\em Euclidean stretch factor} are defined in a similar way, but for the distance measure they use the total Euclidean length of path edges. 
Both hop spanners and Euclidean spanners have received considerable attention in computational geometry and wireless ad hoc networks; see e.g. the surveys by Eppstein \cite{Eppstein1996}, Bose and Smid \cite{Bose2013}, Li \cite{Li2003}, and the book by Narasimhan and Smid \cite{Narasimhan2007}. 
Unit disk graph spanners have been used to reduce the size of a network and the amount of routing information. They are also used in topology control for maintaining network connectivity, improving throughput, and optimizing network lifetime; see the surveys by Li \cite{Li2003} and Rajaraman \cite{Rajaraman2002}. Constructions of UDG spanners, both centralized and
distributed, also with additional properties like planarity and power saving have been widely studied \cite{Gao2005,Li2002,Li2004,Kanj2008}. Researchers also studied the construction of spanners for general disk graphs \cite{Furer2012} and for quasi unit disk graphs \cite{Chen2011}.

\subsection{Related Work}
In this section we review some previous attempts towards getting plane hop spanners for unit disk graphs.
Gao~\etal~\cite{Gao2005} proposed a randomized algorithm that constructs a spanner with constant Euclidean and hop stretch factors. They use a hierarchical clustering algorithm of \cite{Gao2003} to create several clusters of points each containing a point as the clusterhead. Then they connect the clusters by a restricted Delaunay graph
, and then connect the remaining points to clusterheads. The restricted Delaunay graph can be maintained
in a distributed manner when points move
around. Although the underlying restricted Delaunay graph is plane, the entire spanner is not. This spanner has constant Euclidean stretch factor in expectation, and 
constant hop stretch factor for some unspecified constant.

Alzoubi~\etal~\cite{Alzoubi2003} proposed a distributed algorithm for the construction of a hop spanner for the UDG. Their algorithm integrates the connected dominating set and the local Delaunay graph of \cite{Li2002} to form a backbone for the spanner. Although the backbone is plane, the entire spanner is not. The hop stretch factor of this spanner is at most $15716$ (around $15000$ as estimated in \cite{Catusse2010}).

To the best of our knowledge, the only construction that guarantees the planarity of the entire hop spanner is due to Catusse, Chepoi, and Vax{\`{e}}s \cite{Catusse2010}. First they use a regular square-grid to partition input points into clusters. Then they add edges between points in different clusters, and also between points in the same cluster to obtain a hop spanner, which is not necessarily plane. Then they go through several steps and in each step they remove some edges to ensure planarity, and add some new edges to maintain constant hop stretch factor. At the end they obtain a plane hop spanner with hop stretch factor at most $449$. This spanner can be obtained by a localized distributed algorithm. 

\subsection{Our Contribution}

Our main contribution in this paper is a polynomial-time simple algorithm that constructs a plane hop spanner, with hop stretch factor at most $\ratio$, for unit disk graphs. Our algorithm works as follows: Given a set $P$ of points in the plane, we first select a subset $S$ of $P$ (in a clever way), then compute a plane graph $\DT{S}$ (which is the Delaunay triangulation of $S$ minus edges of length more than 1), and then connect every remaining point of $P$ to its closest visible vertex of $\DT{S}$. 
In addition to improving the hop stretch factor, this algorithm is straightforward and the planarity proof is simple, in contrast to that of Catusse \etal~\cite{Catusse2010}. Our analysis of hop stretch factor is still rather involved. Towards the correctness proof of our algorithm, we prove several results on the plane geometry, which are of independent interest. Our construction uses only local information and can be implemented
as a localized distributed algorithm.



Catusse \etal~\cite{Catusse2010} also showed a simple construction of a hop spanner, with hop stretch factor $5$ and with at most $10n$ edges, for any $n$-vertex unit disk graph. With a simple modification to their construction we obtain such a spanner with at most $9n$ edges.

\section{Preliminaries and Some Geometric Results}

We say that a set of points in the plane is in {\em general position} if no three points lie on a straight line and no four points lie on a circle. Throughout this paper, every given point set is assumed to be in general position. For a set $P$ of points in the plane, we denote by $\DelT{P}$ the Delaunay triangulation of $P$. Let $p$ and $q$ be any two points in the plane. We denote by $pq$ the straight-line segment between $p$ and $q$, and by $\ray{pq}$ the ray that emanates from $p$ and passes through $q$. The {\em diametral disk} $D(p,q)$ between $p$ and $q$ is the disk with diameter $|pq|$ that has $p$ and $q$ on its boundary. Every disk considered in this paper is closed, i.e., the disk contains its boundary circle.

Consider the Delaunay triangulation of a point set $P$. In their seminal work, Dobkin, Friedman, and Supowit \cite{Dobkin1990} proved that for any two points $p,q\in P$ there exists a path, between $p$ and $q$ in $\DelT{P}$, that lies in the diametral disk between $p$ and $q$. Their proof makes use of Voronoi cells (of the Voronoi diagram of $P$) that intersect the line segment $pq$. In the following theorem we give a simple inductive proof for a more general claim that shows the existence of such a path in any disk (not only the diametral disk) between $p$ and $q$.

\begin{figure}[htb]
	\centering
	\includegraphics[width=.27\columnwidth]{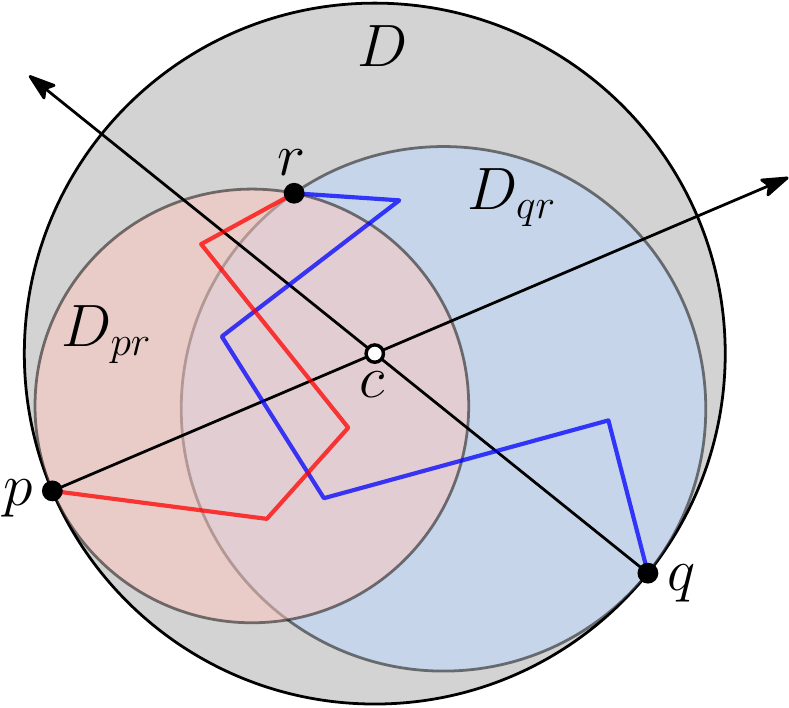}
	\caption{Illustration of the proof of Theorem~\ref{Delaunay-thr}.}
	\label{Delaunay-fig}
\end{figure}

\begin{theorem}
	\label{Delaunay-thr}
	Let $P$ be a set of points in the plane in general position and let $\DelT{P}$ be the Delaunay triangulation of $P$. Let $p$ and $q$ be any two points of $P$ and let $D$ be any disk that has $p$ and $q$ on its boundary. There exists a path, between $p$ and $q$ in $\DelT{P}$, that lies in $D$.
\end{theorem}

\begin{proof}
We prove this lemma by induction on the number of points in $D$. If $D$ does not contain any point of $P\setminus\{p,q\}$ in its interior then $(p,q)$ is an edge of $\DelT{P}$, and thus $(p,q)$ is a desired path. Assume that $D$ contains a point $r\in P\setminus\{p,q\}$ in its interior. Let $c$ be the center of $D$. Consider the ray $\ray{pc}$. Fix $D$ at $p$ and shrink it along $\ray{pc}$ until $r$ becomes on its boundary circle; see Figure \ref{Delaunay-fig}. Denote the resulting disks by $D_{pr}$; this disk lies fully in $D$. Compute the disk $D_{qr}$ in a similar fashion by shrinking $D$ along $\ray{qc}$. Since $r$ is in the interior of $D$, the disk $D_{pr}$ does not contain $q$ and the disk $D_{qr}$ does not contain $p$. Thus, the number of points in each of $D_{pr}$ and $D_{qr}$ is smaller than that of $D$. Therefore, by induction hypothesis there exists a path, between $p$ and $r$ in $\DelT{P}$, that lies in $D_{pr}$, and similarly there exists a path, between $q$ and $r$ in $\DelT{P}$, that lies in $D_{qr}$. The union of these two paths contains a path, between $p$ and $q$ in $\DelT{P}$, that lies in $D$.
\end{proof}

Let $G$ be a plane geometric graph and let $p\notin G$ be any point in the plane. We say that a vertex $q\in G$ is {\em visible} from $p$ if the straight-line segment $pq$ does not cross any edge of $G$. One can simply verify that for every $p$ such a vertex $q$ exists. Among all vertices of $G$ that are visible from $p$, we refer to the one that is closest to $p$ by the {\em closest visible vertex} of $G$ from $p$.

The following theorem (though simple) turns out to be crucial in the planarity proof of our hop spanner; this theorem is of independent interest. Although it answers a basic question, we were unable to find such a result in the literature; there exist however related results, see e.g. \cite{Akitaya2015,Bose2012,Hurtado2008}.

\begin{figure}[htb]
	\centering
	\includegraphics[width=.28\columnwidth]{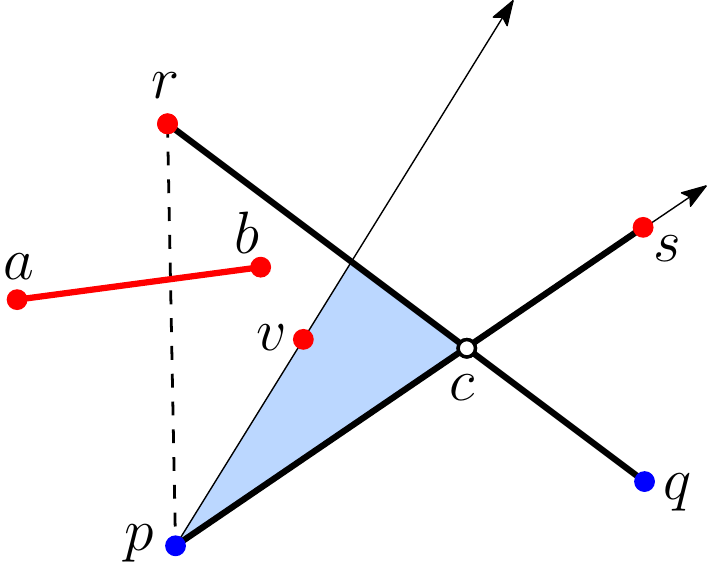}
	\caption{Illustration of the proof of Theorem~\ref{pslg-thr}. The red vertices belong to $G$.}
	\label{pslg-fig}
\end{figure}

\begin{theorem}
	\label{pslg-thr}
	Let $G$ be a plane geometric graph, and let $Q$ be a set of points in the plane that is disjoint from $G$. The graph, that is obtained by connecting every point of $Q$ to its closest visible vertex of $G$, is plane. 
\end{theorem}
\begin{proof}
Let $E$ be the set of edges that connect every point of $Q$ to its closest visible vertex of $G$. To prove the theorem, it suffices to show that the edges of $G\cup E$ do not cross each other. The edges of $G$ do not cross each other because $G$ is plane. It is implied from the definition of visibility that the edges of $E$ do not cross the edges of $G$. 

It remains to prove that the edges of $E$ do not cross each other. We prove this by contradiction. To that end consider two crossing edges $(p,s)$ and $(q,r)$ of $E$ where $p,q$ are two points of $Q$ and $s,r$ are two vertices of $G$. Let $c$ be their intersection point of $(p,s)$ and $(q,r)$. By the triangle inequality we have $|pr|<|ps|$ or $|qs|<|qr|$. After a suitable relabeling assume that $|pr|<|ps|$, and thus $p$ is closer to $r$ than to $s$. The reason that $p$ was not connected to $r$, is that $r$ is not visible from $p$. Therefore there are edges of $G$ that block the visibility of $r$ from $p$. Take any such edge $(a,b)$. The edge $(a,b)$ does not intersect any of $(p,s)$ and $(q,r)$ because otherwise $(a,b)$ blocks the visibility of $s$ from $p$ or the visibility of $r$ from $q$, and as such we wouldn't have these edges in $E$; see Figure~\ref{pslg-fig}. Therefore, exactly one endpoint of $(a,b)$, say $b$, lies in the triangle $\bigtriangleup pcr$. Rotate the ray $\ray{ps}$ towards $b$ and stop as soon as hitting a vertex of $G$ in $\bigtriangleup pcr$. This vertex is visible from $p$. Denote this vertex by $v$ (it might be that $v=b$). Since $v$ lies in $\bigtriangleup pcr$, it turns out that $|pv|\leqslant\max\{|pr|,|pc|\}<|ps|$. Thus, $v$ is a closer visible vertex of $G$ from $p$. This contradicts the fact that $s$ is a closest visible vertex from $p$. 
\end{proof}

\begin{lemma}
	\label{Convex-lemma}
	Let $C$ be a convex shape of diameter $d$ in the plane, and let $pq$ be a straight-line segment that intersects $C$. Then the distance from any point $r\in C$ to $p$ or to $q$ is at most $\sqrt{d^2+|pq|^2/4}$.
\end{lemma}

\begin{proof}
Let $s$ be a point in the intersection of $C$ and $pq$. Let $l_{pq}$ be the line through $pq$, and let $c$ be the point of $l_{pq}$ that is closest to $r$. Observe that $\bigtriangleup rsc$ is a right triangle with hypotenuse $rs$, and thus $|rc|\leqslant |rs|\leqslant d$. If an endpoint of $pq$ lies on segment $sc$ (as depicted in Figure~\ref{convex-fig}) then the distance from $r$ to that endpoint is at most $d$. Assume that no endpoint of $pq$ lies on $sc$, and thus $c$ lies on $pq$. After a suitable relabeling assume that $c$ is closer to $p$ than to $q$, and thus $|cp|\leqslant |pq|/2$. In this setting, by the Pythagorean equation we get $|rp|=\sqrt{|rc|^2+|cp|^2}\leqslant \sqrt{d^2+|pq|^2/4}$.
\end{proof}

\begin{figure}[htb]
	\centering
	\includegraphics[width=.31\columnwidth]{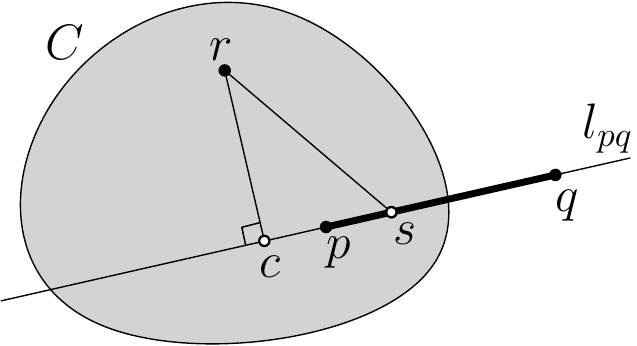}
	\caption{Illustration of the proof of Lemma~\ref{Convex-lemma}}
	\label{convex-fig}
\end{figure}

We refer to a hop spanner with hop stretch factor $t$ as a $t${\em-hop spanner}.
Catusse \etal~\cite{Catusse2010} showed a simple construction of a sparse 5-hop spanner with at most $10n$ edges, for any $n$-vertex unit disk graph. With a simple modification to their construction we obtain a 5-hop spanner with at most $9n$ edges. 

\begin{figure}[htb]
	\centering
	\includegraphics[width=.23\columnwidth]{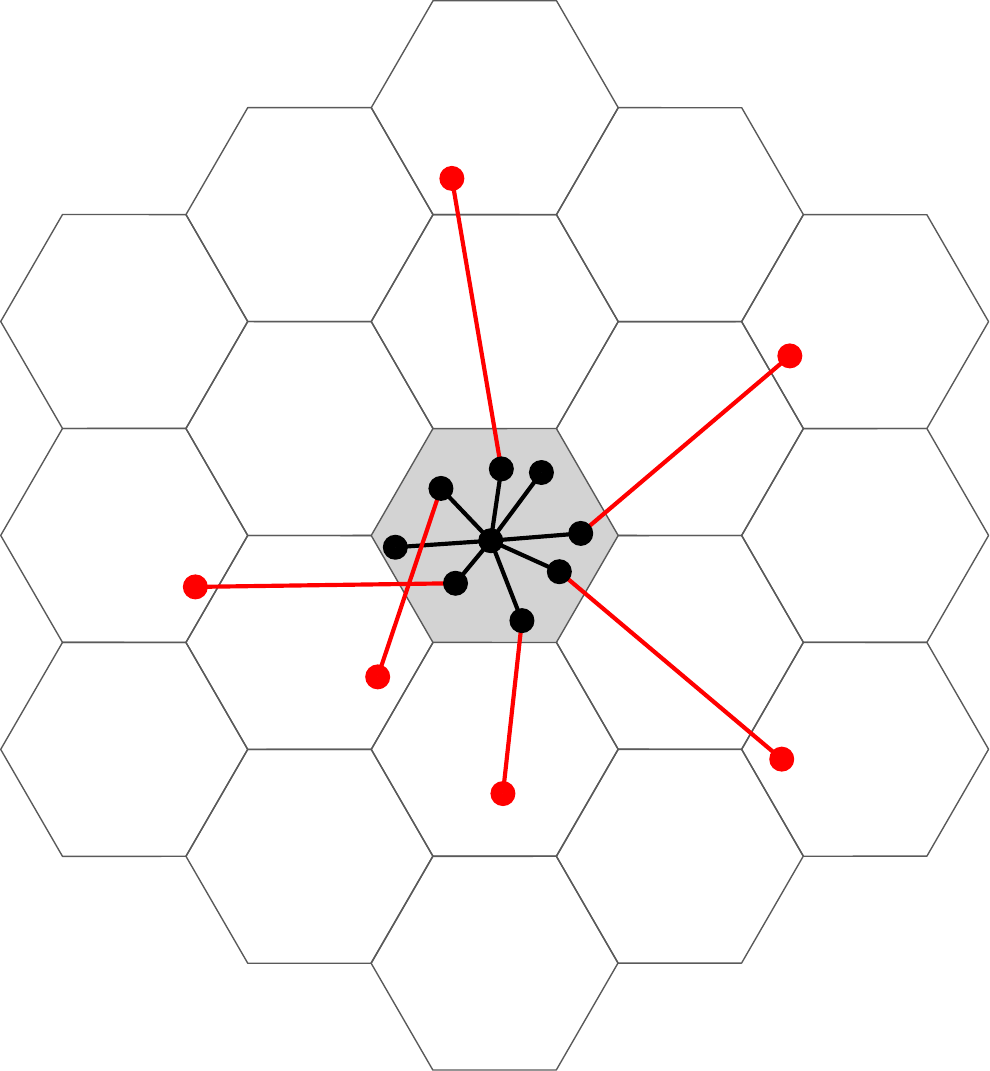}
	\caption{Every cell can have edges to at most $18$ other cells.}
	\label{sparse-spanner-fig}
\end{figure}

\begin{theorem}
	\label{sparse-spanner-thr}
	Every $n$-vertex unit disk graph, has a $5$-hop spanner with at most $9n$ edges.
\end{theorem}

\begin{proof}
Consider the unit disk graph $\UDG{P}$ on any set $P$ of $n$ points in the plane. 
Consider a regular hex-grid on the plane with hexagons (cells) of diameter 1.
In every nonempty cell $\pi$ pick a point as the {\em center} and connect it to all other points in this cell; these edges are in $\UDG{P}$ because the diameter of $\pi$ is $1$. Then take exactly one edge of $\UDG{P}$ between any two cells if such an edge exists; from each  cell we can have edges to at most $18$ other cells as depicted in the Figure~\ref{sparse-spanner-fig} (Catusse~\etal~\cite{Catusse2010} use a square-grid in which every cell can have edges to at most $20$ other cells). We claim that the resulting graph, which we call it $H$, is a desired spanner. By a counting argument one can verify that $H$ has at most $9n$ edges. To verify the hop stretch factor consider two points $p,q\in P$. If $p$ and $q$ are in the same cell, then there is a path, of length at most $2$ between $p$ and $q$ in $H$, that goes through the center of the cell. Assume that $p$ and $q$ lie in different cells, say $\pi(p)$ and $\pi(q)$. By our construction there is an edge, say $(p',q')$, in $H$ between $\pi(p)$ and $\pi(q)$. Therefore, there is a path of length at most $5$ between $p$ and $q$, that goes through $p',q'$, and through the centers of $\pi(p)$ and $\pi(q)$. 
\end{proof}

\section{Plane Hop Spanner Algorithm}
This section presents our main contribution which is a polynomial-time algorithm for construction of plane hop spanners for unit disk graphs.

Let $P$ be a set of points in the plane in general position, and let $\UDG{P}$ be the unit disk graph of $P$. 
Our algorithm first partitions $P$ into some clusters by using a regular square-grid; this is a standard initial step in many UDG algorithms, see e.g. \cite{Catusse2010,Chen2011}. We use this partition to select a subset $S$ of $P$ that satisfies some properties, which we will describe later. Then we compute the Delaunay triangulation of $S$ and remove every edge that has length more than $1$. We denote the resulting graph by $\DT{S}$. Then we connect every point of $P\setminus S$ to its closest visible vertex of $\DT{S}$. Let $H$ denote the final resulting graph. We claim that $H$ is a plane hop spanner, with hop stretch factor at most $\ratio$, for $\UDG{P}$. In Section~\ref{S-section} we show how to compute $S$. 
The points of $S$ are distributed with constant density, i.e., there are $O(1)$ points of $S$ in any unit disk in the plane. Based on this and the fact that $\DT{S}$ has only edges of length at most 1, $\DT{S}$ can be computed by a localized distributed algorithm.
In Section~\ref{correctness-section} we prove the correctness of the algorithm that $H$ is plane and $H$ is a subgraph of $\UDG{P}$. In Section~\ref{stretch-section} we analyze the stretch factor of $H$. The following theorem summarizes our result in this section.

\begin{theorem}
	There exists a plane $\ratio$-hop spanner for the unit disk graph of any set of points in the plane in general position. Such a spanner can be computed in polynomial time. 
\end{theorem}

\subsection{Computation of $S$}
\label{S-section}
In this section, we compute the subset $S$; we will see properties of $S$ at the end of this section.
Let $\Gamma$ be a regular square-grid on the plane with squares of diameter 1. The side-length of these squares is $1/\sqrt{2}$. Without loss of generality we assume that no point of $P$ lies on a grid line (this can be achieved by moving the grid by a small amount horizontally and vertically).
Let $E$ be the edge set containing the shortest edge of $\UDG{P}$ that runs between any two nonempty cells of $\Gamma$ if such an edge exists. Since every edge of $\UDG{P}$ has length at most 1, for every cell $\pi$ there are at most 20 edges in $E$ going from $\pi$ to other cells $\pi_1,\dots,\pi_{20}$ as depicted in Figure~\ref{grid-fig}
. 
Let $V(E)$ be the set of endpoints of $E$, i.e., endpoints of the edges of $E$. 
The set $V(E)$ has the following two properties:

\begin{itemize}
	\item Every cell of $\Gamma$ contains at most 20 points of $V(E)$.
	\item For every cell $\pi\in \Gamma$ and every $i\in\{1,\dots,20\}$ if there is an edge in $\UDG{P}$ between $\pi$ and $\pi_i$, then there are two points $s_i,t_i\in V(E)$ such that $s_i\in \pi$, $t_i\in\pi_i$, and $(s_i,t_i)$ is the shortest edge of $\UDG{P}$ that runs between $\pi$ and $\pi_i$. 
\end{itemize}

\begin{figure}
	\centering
	\includegraphics[width=.7\columnwidth]{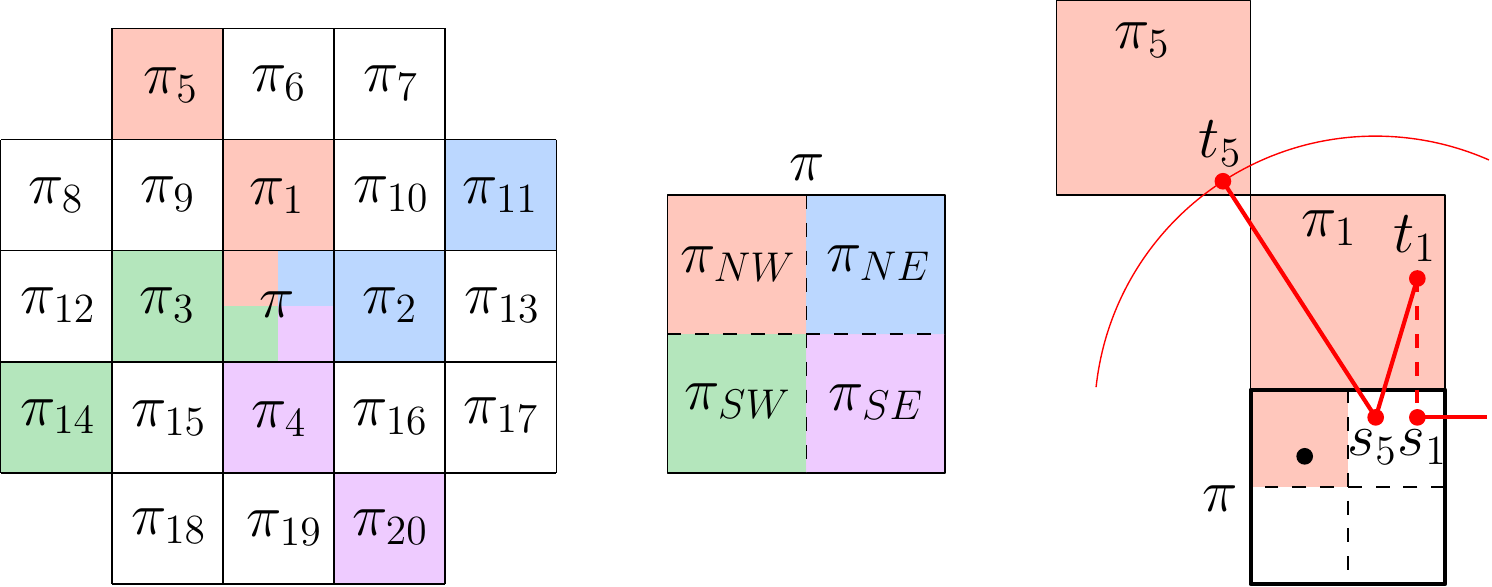}
	\caption{Illustration of the computation of $S$.}
	\label{grid-fig}
\end{figure}

We want to modify the edge set $E$ and also compute a point set $T$ such that $V(E)\cup T$ satisfies some more properties that we will see later. To that end we partition every cell $\pi$ of $\Gamma$ into four sub-cells of diameter $1/2$, namely $\pi_{NW}, \pi_{NE},\pi_{SW}, \pi_{SE}$ as in Figure~\ref{grid-fig}. 
For each cell $\pi$, consider four triplets $(\pi_{NW},\allowbreak \pi_1,\pi_5)$, $(\pi_{NE},\pi_2,\pi_{11})$, $(\pi_{SW},\allowbreak \pi_3,\pi_{14})$, and $(\pi_{SE},\pi_4,\pi_{20})$; these triplets are colored in Figure~\ref{grid-fig}. Let $T$ be the empty set. We perform the following three-step process on each of the four triplets of every cell $\pi$. 
We describe the process only for $(\pi_{NW},\pi_1,\pi_5)$; the processes of other triplets are analogous. In our description ``a point of $\pi_{NW}$'' refers to a point of $P$ that lies in $\pi_{NW}$. 


\begin{enumerate}
	\item[1.] If $\pi_{NW}$ is empty (contains no point of $P$) then we do nothing and stop the process. Assume that $\pi_{NW}$ contains some points of $P$. If $\pi_{NW}$ contains an endpoint of $E$, i.e. an endpoint of some edge of $E$, then we do nothing and stop the process. 
	\item[2.] Assume now that $\pi_{NW}$ contains some points of $P$ but does not contain any endpoint of $E$. If there is no edge in $E$ that runs between $\pi$ and $\pi_1$ or between $\pi$ and $\pi_5$ then we take a point of $\pi_{NW}$ arbitrary and add it to $T$, and then stop the process.  
	\item[3.] Assume that $E$ contains an edge between $\pi$ and $\pi_1$, and an edge between $\pi$ and $\pi_5$. We are now in the case where $\pi_{NW}$ contains some points of $P$ but not any endpoint of $E$, and both $s_1$ and $s_5$ exist. If $s_1=s_5$, then we add a point of $\pi_{NW}$ to $T$, and then stop the process. Assume that $s_1\neq s_5$. Since $\pi_{NW}$ does not contain any endpoint of $E$, the points $s_1$ and $s_5$ do not lie in $\pi_{NW}$. In particular, $s_5$ lies in sub-cell $\pi_{NE}$ because the distance between $\pi_5$ and each of $\pi_{SE}$ and $\pi_{SW}$ is more than 1; however $s_1$ might lie in other sub-cells. In this setting the disk with center $s_5$ and radius $|s_5t_5|$ contains the entire $\pi_1$ (see Figure~\ref{grid-fig}), and thus the distance between $s_5$ and any point in $\pi_1$ is at most $1$. We replace the edge $(s_1,t_1)$ of $E$ by the edge $(s_5,t_1)$, which has length at most one; see Figure~\ref{grid-fig}. Then we add a point of $\pi_{NW}$ to $T$, and then stop the process. 
\end{enumerate}

This is the end of process for triplet $(\pi_{NW},\pi_1,\pi_5)$. After performing this process on all triplets of all cells, we obtain an edge set $E$ and a point set $T$. We define the subset $S$ to be union of $T$ and the endpoints of edges of $E$, i.e., $S=V(E)\cup T$. We will use the properties in the following lemma in correctness proof and analysis of hop stretch factor.

\begin{lemma}
	\label{S-lemma}
	The set $S$ satisfies the following four properties:
	\begin{enumerate}
		\item[$(\textup{P}1)$] Every cell of $\Gamma$ contains at most 20 points of $S$.
		\item[$(\textup{P}2)$] For every cell $\pi$ and every $i\in\{1,2,3,4\}$, if there is an edge in $\UDG{P}$ between $\pi$ and $\pi_i$, then there are two points $s_i,t_i\in S$ such that $s_i\in \pi$, $t_i\in\pi_i$, and $|s_it_i|\leqslant 1$.
		\item[$(\textup{P}3)$] For every cell $\pi$ and every $i\in\{5,\dots,20\}$, if there is an edge in $\UDG{P}$ between $\pi$ and $\pi_i$, then there are two points $s_i,t_i\in S$ such that $s_i\in \pi$, $t_i\in\pi_i$, $|s_it_i|\leqslant 1$, and $(s_i,t_i)$ is the shortest edge of $\UDG{P}$ that runs between $\pi$ and $\pi_i$.
		\item[$(\textup{P}4)$] The set $S$ contains at least one point from every nonempty sub-cell $\pi_{NW}$, $\pi_{NE}$, $\pi_{SW}$, $\pi_{SE}$ of each cell $\pi$. 	
	\end{enumerate} 
\end{lemma}
\begin{proof}
Recall that $E$ initially contains shortest edges between different cells. In step 3 we replace only the edges of $E$, that run between each cell $\pi$ and the cells $\pi_1$, $\pi_2$, $\pi_3$, $\pi_4$, with new edges of length at most $1$. Therefore properties (P2) and (P3) hold. Every nonempty sub-cell contains either a point in $V(E)$ (step 1) or a point in $T$ (steps 2 and 3), and thus property (P4) holds. 

To verify property (P1), we use the discharging method as follows. Consider one cell $\pi$.
Before the process, we give charge $1$ to each $\pi_i$ for $i\in\{1,\dots,20\}$. Thus, the total available charge for $\pi$ is 20. Then, for every edge $(s_i,t_i)\in E$ we move the charge of $\pi_i$ to $s_i$. Since each $s_i$ can be an endpoint of more than one edge of $E$, it may get charges from more than one cell. During the process we move charges as follows. In step 2, if there is no edge in $E$ that runs between $\pi$ and $\pi_1$ or between $\pi$ and $\pi_5$, then we move the charge of $\pi_1$ or $\pi_5$ to the point of $\pi_{NW}$ that we add to $T$, respectively. 
Now consider step 3. If $s_1=s_5$ then $s_1$ has charge at least 2 that come from $\pi_1$ and $\pi_5$. In this case we move charge 1 from $s_1$ to the point of $\pi_{NW}$ that we add to $T$. If $s_1\neq s_5$, then after replacing $(s_1,t_1)$ with $(s_5,t_1)$, we move the charge of $\pi_1$ from $s_1$ to the point of $\pi_{NW}$ that we add to $T$. 
After this replacement if $s_1$ is not an endpoint of any edge of $E$ other than $(s_1,t_1)$ then $s_1$ gets removed from $V(E)$, otherwise it still holds charges of some cells other than $\pi_1$. Thus, after processing $\pi$, the final charge of every point of $S$, that lies in $\pi$, is at least $1$. Observe that $\pi_1$ and $\pi_5$ belong to only one of the four triplets that are associated to $\pi$, and thus we do not double count their charges. Since the total available charge for $\pi$ was $20$, it turns out that the number of points of $S$, that lie in $\pi$, is at most $20$.

While processing each cell $\pi$, we add to $T$ only points that lie in $\pi$. Moreover, the edge-replacement of step 3, does not add any new point to $V(E)$. After processing $\pi$, there is an edge in $E$ running between $\pi$ and $\pi_i$ if and only if there was such an edge before processing $\pi$. Therefore, after processing all cells, every cell contains at most $20$ points of $S$, and thus (P1) holds.
\end{proof}

\subsection{Correctness Proof}
\label{correctness-section}
In this section we prove the correctness of our algorithm. Recall the grid $\Gamma$, and the subset $S$ of $P$ that is computed in Section~\ref{S-section}. Recall that our algorithm computes the Delaunay triangulation $\DelT{S}$ and removes every edge of length more than $1$ to obtain $\DT{S}$, and then connects every point of $P\setminus S$ to its closest visible vertex of $\DT{S}$. Let $H$ denotes the resulting graph. One can simply verify that this algorithm takes polynomial time. Since $\DelT{S}$ is plane, its subgraph $\DT{S}$ is also plane. It is implied from Theorem~\ref{pslg-thr} (where $\DT{S}$ and $P\setminus S$ play the roles of $G$ and $Q$) that $H$ is plane. As we stated at the outset, except for the computation of $S$ which is a little more involved, the algorithm and the planarity proof are straightforward.

To finish the correctness proof it remains to show that every edge of $H$ has length at most $1$. Consider any edge $e$ of $H$. 
By our construction, either the two endpoints of $e$ belong to $S$, or one endpoint of $e$ belongs to $S$ and its other endpoint belongs to $P\setminus S$. If both endpoints of $e$ are in $S$, then $e$ belongs to $\DT{S}$ and hence has length at most 1. If one endpoint of $e$ is in $S$ and its other endpoint is in $P\setminus S$, then by following lemma the length of $e$ is at most $1/\sqrt{2}$.

\begin{figure}[htb]
	\centering
	\includegraphics[width=.28\columnwidth]{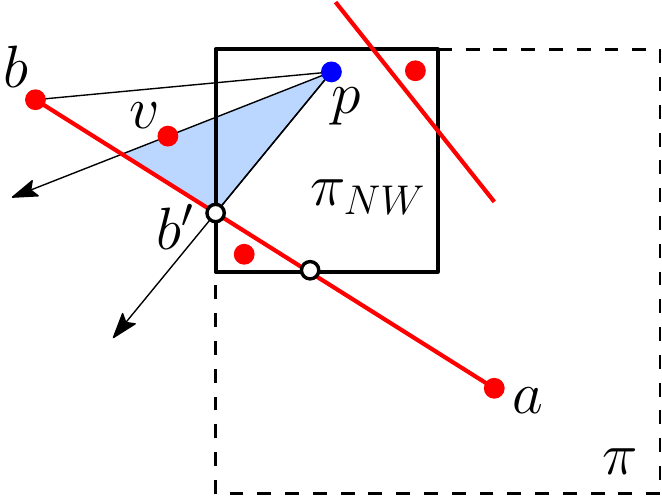}
	\caption{Illustration of the proof of Lemma~\ref{unselected-length-lemma}. The red points belong to $S$.}
	\label{unselected-fig}
\end{figure}

\begin{lemma}
	\label{unselected-length-lemma}
	The length of every edge of $H$, that has an endpoint in $S$ and an endpoint in $P\setminus S$, is at most $1/\sqrt{2}$.
\end{lemma}

\begin{proof}	
Consider any edge $(p,s)\in H$ with $p\in P\setminus S$ and $s\in S$. By our construction, $s$ is the closest visible vertex of $\DT{S}$ from $p$. Thus, to prove the lemma, it suffices to show the existence of a vertex $v\in \DT{S}$ that is visible from $p$ and for which $|pv|\leqslant 1/\sqrt{2}$; this would imply that the distance between $p$ and $s$, which is the closest visible vertex from $p$, is at most $1/\sqrt{2}$. In the rest of the proof we show the existence of such vertex $v$.

Let $\pi$ be the cell that contains $p$ (the dashed cell in Figure~\ref{unselected-fig}). After a suitable rotation we assume that $p$ lies in sub-cell $\pi_{NW}$. Since $\pi_{NW}$ is nonempty, by property (P4) in Lemma~\ref{S-lemma} the set $S$ contains at least one point from $\pi_{NW}$. Let $S'$ be the set of points of $\pi_{NW}$ that are in $S$. Notice that $S'\subseteq S$ and $S'\neq \emptyset$. If any point of $S'$ is visible from $p$, then this point is a desired vertex $v$ with $|pv|\leqslant 1/2$ because the diameter of $\pi_{NW}$ is $1/2$.

Assume that no point of $S'$ is visible from $p$. 
The visibility of (points of) $S'$ from $p$ is blocked by some edges of $\DT{S}$; these edges properly cross $\pi_{NW}$ and separate $p$ from points of $S'$ (the red edges in Figure~\ref{unselected-fig}). Among these edges take one whose intersection points with the boundary of $\pi_{NW}$ are visible from $p$ (observe that such an edge always exists). 
Denote this edge by $(a,b)$. Since the diameter of $\pi_{NW}$ is $1/2$ and $|ab|\leqslant 1$, it is implied from Lemma~\ref{Convex-lemma} that the distance from $p$ to $a$ or to $b$ is at most $1/\sqrt{2}$; after a suitable relabeling assume that $|pb|\leqslant 1/\sqrt{2}$. Of the two intersection points of $(a,b)$ with the boundary of $\pi_{NW}$, denote by $b'$ the one that is closer to $b$. By our choice of $(a,b)$, $b'$ is visible from $p$. We rotate the ray $\ray{pb'}$ towards $b$ and stop as soon as hitting a vertex $v\in S$ in triangle $\bigtriangleup pbb'$ (it might be that $v=b$). The vertex $v$ is visible from $p$. Since $v$ is in triangle $\bigtriangleup pbb'$ it holds that $|pv|\leqslant \max\{|pb|,|pb'|\}$. Since $|pb|\leqslant 1/\sqrt{2}$ and $|pb'|\leqslant 1/2$ it turns out that $|pv|\leqslant 1/\sqrt{2}$.   
\end{proof}

\subsection{Hop Stretch Factor}
\label{stretch-section}
In this section we prove that the hop stretch factor of $H$ is at most $\ratio$. We show that for any edge $(u,v)\in\UDG{P}$ there exists a path of length at most $\ratio$ between $u$ and $v$ in $H$. 

In this section a ``cell'' refers to the interior of a square of $\Gamma$, a ``grid point'' refers to the intersection point of a vertical and a horizontal grid line, and a ``corner of $\pi$'' refers to a grid point on the boundary of a cell $\pi$. We define {\em neighbors} of a cell $\pi$ to be the set of eight cells that share sides or corners with $\pi$. We partition the neighbors of $\pi$ into $+$-{\em neighbors} and $\times$-{\em neighbors}, where $+$-neighbors are the four cells that share sides with $\pi$, and $\times$-neighbors are the four cells each sharing exactly one grid point with $\pi$. In Figure~\ref{grid-fig} the cells $\pi_1,\pi_2,\pi_3,\pi_4$ are the $+$-neighbors of $\pi$, and the cells $\pi_9,\pi_{10},\pi_{15},\pi_{16}$ are the $\times$-neighbors of $\pi$.

Consider any two points $p,q\in S$. If $|pq|\leqslant 1$ then every edge of $\DelT{S}$, that lies in $D(p,q)$, has length at most 1, and thus all these edges are present in $\DT{S}$. Combining this with Theorem~\ref{Delaunay-thr} we get the following corollary.

\begin{corollary}
	\label{Delaunay-cor}
	For any two points $p,q\in S$, with $|pq|\leqslant 1$, there exists a path, between $p$ and $q$ in $\DT{S}$, that lies in $D(p,q)$. 
\end{corollary}

\begin{figure}[htb]
	\centering
	\setlength{\tabcolsep}{0in}
	$\begin{tabular}{cc}
	\multicolumn{1}{m{.5\columnwidth}}{\centering\includegraphics[width=.3\columnwidth]{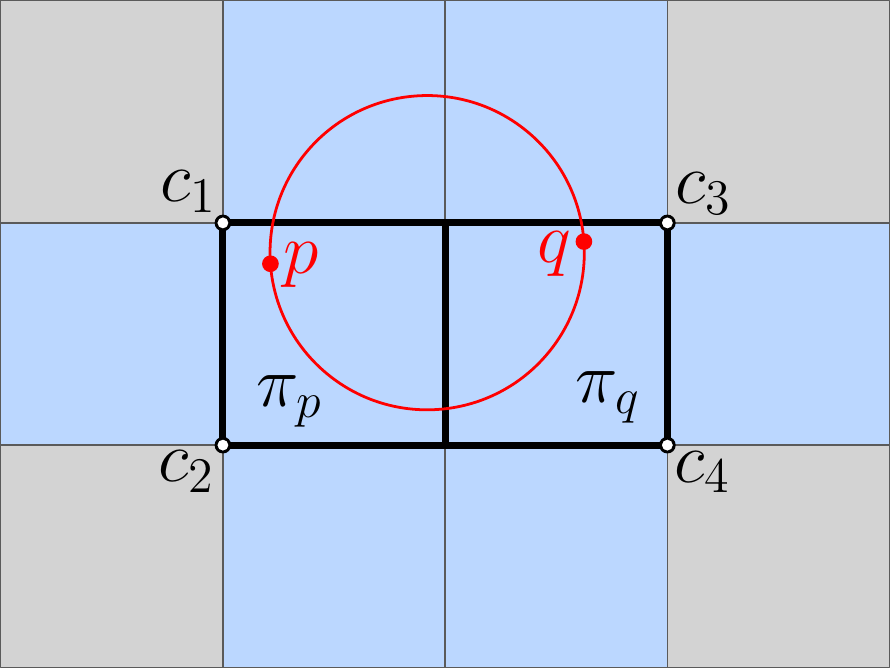}}
	&\multicolumn{1}{m{.5\columnwidth}}{\centering\includegraphics[width=.3\columnwidth]{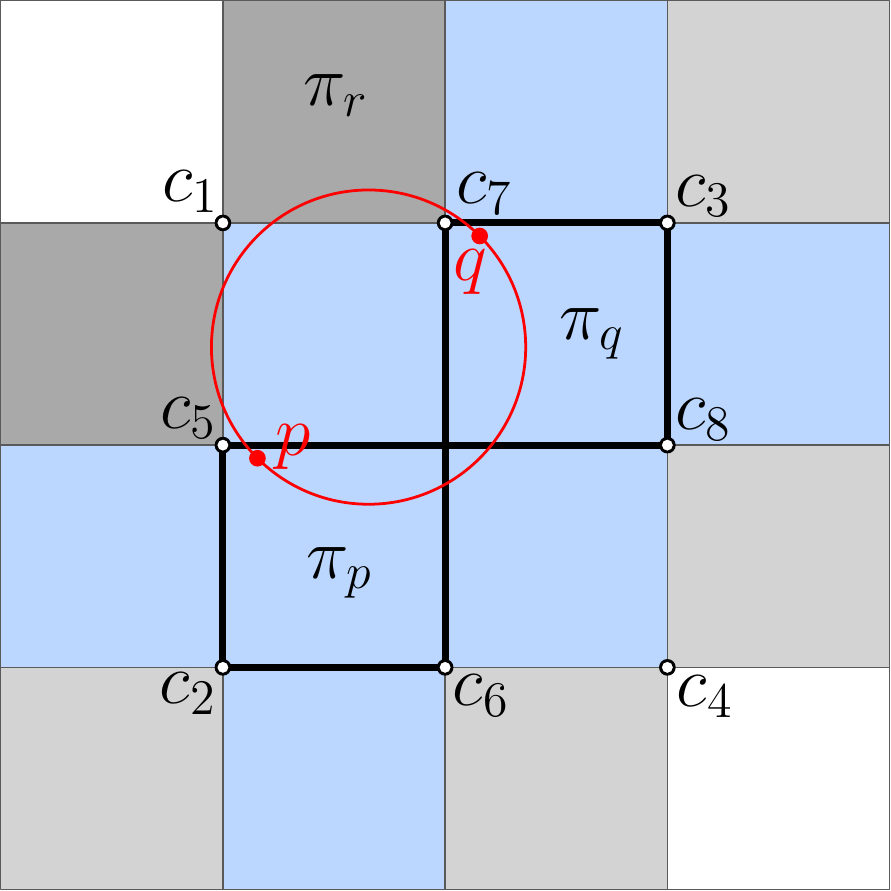}}\\
	(a) Configuration A&(b) Configuration B\\
	\multicolumn{1}{m{.5\columnwidth}}{\centering\includegraphics[width=.35\columnwidth]{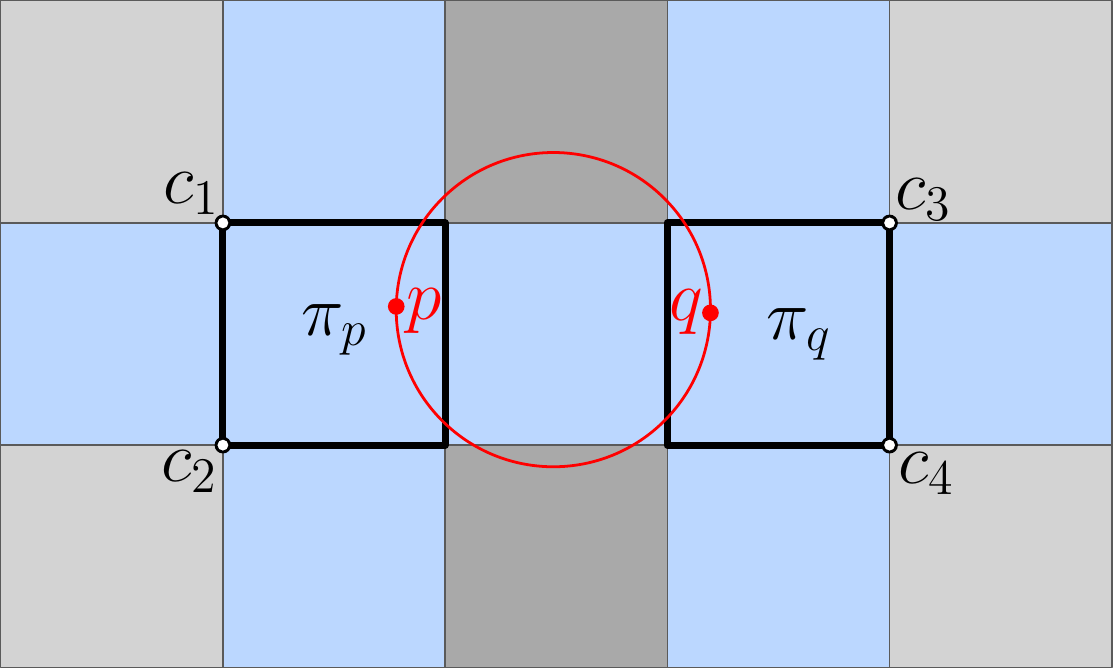}}
	&\multicolumn{1}{m{.5\columnwidth}}{\centering\includegraphics[width=.35\columnwidth]{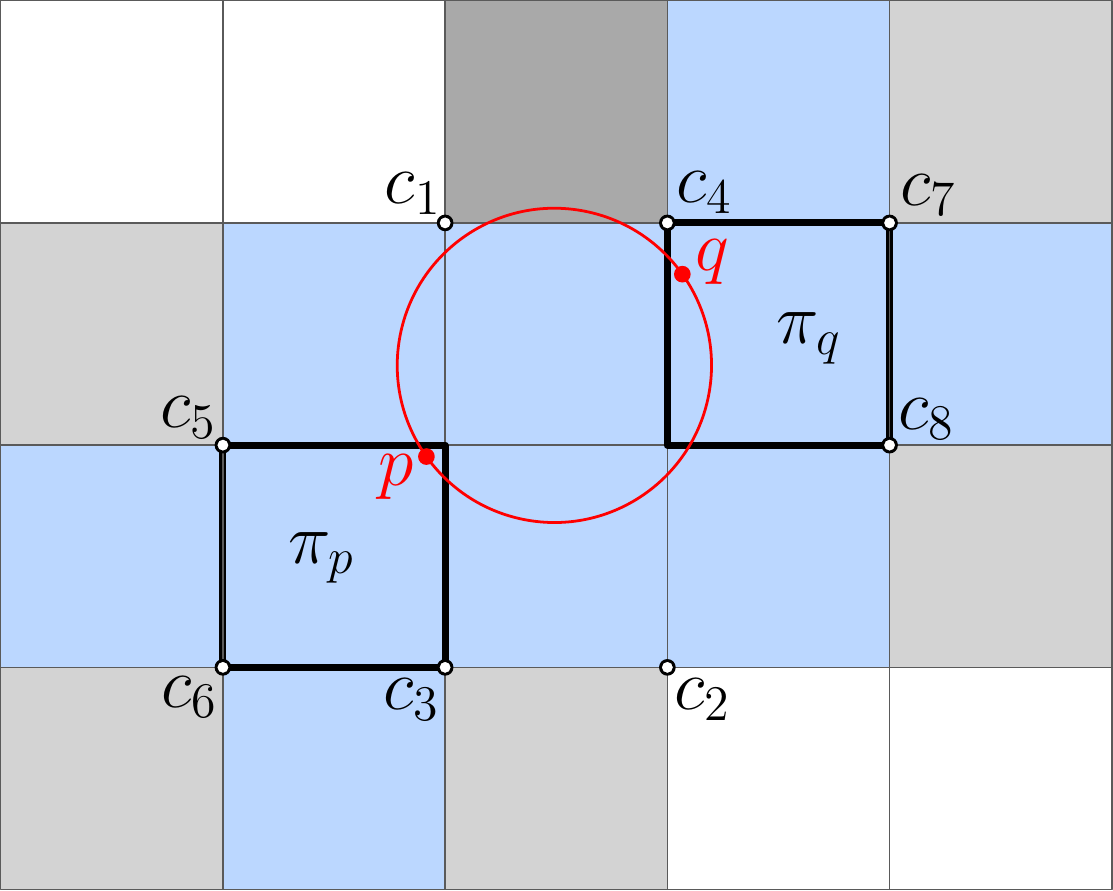}}\\
	(c) Configuration C&(d) Configuration D
	\end{tabular}$
	\caption{Relative positions of the cells $\pi_p$ and $\pi_q$ where $|pq|\leqslant 1$.}
	\label{config-fig}
\end{figure}

Consider any two points $p$ and $q$ in the plane that lie in different cells, say $\pi_p$ and $\pi_q$. If $|pq|\leqslant 1$ then the relative positions of $\pi_p$ and $\pi_q$ is among four configurations A, B, C, and D that are shown in Figure~\ref{config-fig}. In the rest of this section we consider different configurations of a disk intersecting some cells of $\Gamma$. Although mentioned before, we emphasis that a ``cell'' refers to the interior of a square of grid (and hence a cell is open and does not contain its boundary) while a ``disk'' is closed (and hence contains its boundary).

\subsubsection{Disk-Cell Intersections}
To cope with the number of cases that appear in the analysis of hop stretch factor we use lemmas \ref{Circle-Square-lemma}, \ref{two-cell-lemma}, and \ref{neighbor-cell-lemma} about disk-cell intersections. These lemmas enable us to reduce the number of cases in our analysis. We say that an element $x$ is ``outside'' a set $X$ if $x\notin X$. 

\begin{lemma}
	\label{Circle-Square-lemma}
	Let $p$ and $q$ be any two points in the plane with $|pq|\leqslant 1$.
	\begin{enumerate}
		\item[1.] If $p$ and $q$ are in different cells, then $D(p,q)$ intersects at most $7$ cells.
		\item[2.] If $p$ and $q$ are in the same cell $\pi$, then $D(p,q)$ can intersect only $\pi$ and its four $+$-neighbors.
	\end{enumerate}
\end{lemma}
\begin{proof}
Statement 1 is implied by the fact that $D(p,q)$ contains at most two grid points in its interior. To verify statement 2, it suffices to show that $D(p,q)$ does not contain any corner of $\pi$. Consider a corner $c$ of $\pi$. Since $p$ and $q$ lie in $\pi$, the convex angle $\angle pcq$ is acute. Combing this with Thales's theorem implies that $c$ is outside $D(p,q)$.
\end{proof}

\begin{lemma}
	\label{two-cell-lemma}
	Let $p$ and $q$ be any two points in the plane that are in different cells $\pi_p$ and $\pi_q$. Let $X$ be the set containing the cells $\pi_p$ and $\pi_q$ and their $+$-neighbors.
	\begin{enumerate}
		\item[1.] If $|pq|\leqslant 1$, then $D(p,q)$ does not intersect any cell outside the neighborhoods of $\pi_p$ and $\pi_q$.
		\item[2.] If $|pq|\leqslant 1$, then $D(p,q)$ intersects at most two cells outside $X$.
		\item[3.] If $|pq|\leqslant 1/\sqrt{2}$, then $D(p,q)$ does not intersect any cell outside $X$.
	\end{enumerate}
\end{lemma}

\begin{proof}
We prove each statement separately.

\vspace{5pt}
{\em Statement 1.} Any cell $\pi$, that is outside the neighborhoods of $\pi_p$ and $\pi_q$, has distance more $1/\sqrt{2}$ from each of $\pi_p$ and $\pi_q$. Thus, for any point $r\in \pi$ we have $|rp|> 1/\sqrt{2}$ and $|rq|> 1/\sqrt{2}$. Since $|pq|\leqslant 1$, for any point $x$ in $D(p,q)$ it holds that either $|xp|\leqslant 1/\sqrt{2}$ or $|xq|\leqslant 1/\sqrt{2}$. Therefore, $r$ cannot be in $D(p,q)$. This implies that $D(p,q)$ does not intersect $\pi$. 

\vspace{5pt}
{\em Statement 2.} The relative positions of $\pi_p$ and $\pi_q$ is among the four configurations in Figure~\ref{config-fig}. In this figure, the cells of $X$ are colored blue, the $\times$-neighbors of $\pi_p$ and $\pi_q$ that are not in $X$ and not intersected by $D(p,q)$ are colored light gray, and the $\times$-neighbors of $\pi_p$ and $\pi_q$ that are not in $X$ but intersected by $D(p,q)$ are colored dark gray. We prove this statement for each  configuration.

\begin{itemize}
	\item Configuration A. By an application of Thales's theorem as in the proof of Lemma~\ref{Circle-Square-lemma}, one can verify that $D(p,q)$ does not contain any of grid points $c_1$, $c_2$, $c_3$, $c_4$, and hence does not intersect any cell outside $X$; see Figure~\ref{config-fig}(a).
	\item Configuration B. By Thales's theorem, $D(p,q)$ does not contain any of grid points $c_1$, $c_2$, $c_3$, $c_4$; see Figure~\ref{config-fig}(b). The mutual distances between grid points $c_5$, $c_6$, $c_7$, $c_8$ is at least $1$, and thus $D(p,q)$ contains at most one of them. With these constraints, it turns out that $D(p,q)$ intersects at most two cells outside $X$.
	\item Configuration C. By Thales's theorem, $D(p,q)$ does not contain any of grid points $c_1$, $c_2$, $c_3$, $c_4$; see Figure~\ref{config-fig}(c). In this setting, $D(p,q)$ intersects at most two cells outside $X$.
	\item Configuration D. See Figure~\ref{config-fig}(d). Since $|c_1q|>1/\sqrt{2}$ and $|c_1p|>1/\sqrt{2}$, by an argument similar to the proof of statement 1, one can verify that $D(p,q)$ does not contain $c_1$; by symmetry it also does not contain $c_2$. Since the distance between $q$ and each of $c_3$, $c_5$, $c_6$ is more than $1$, $D(p,q)$ does not contain any of $c_3$, $c_5$, and $c_6$; by symmetry it also does not contain any of $c_4$, $c_7$, and $c_8$. With these constraints, it turns out that $D(p,q)$ intersects at most one cell outside $X$. 
\end{itemize} 

{\em Statement 3.} Since $|pq|\leqslant 1/\sqrt{2}$, the cells $\pi_p$ and $\pi_q$ are neighbors, and thus their relative positions is among configurations A and B in Figure~\ref{config-fig}. We have seen in the proof of statement 2 that in configuration A the disk $D(p,q)$ does not intersect any disk outside $X$. 
We prove our claim for configuration B. Every cell $\pi_r$ outside $X$ is at a distance more than $1/\sqrt{2}$ from $\pi_p$ or $\pi_q$; see Figure~\ref{config-fig}(b). Since the diameter of $D(p,q)$ is at most $1/\sqrt{2}$, this implies that $\pi_r$ lies outside $D(p,q)$.
\end{proof}

\begin{figure}[htb]
	\centering
	\setlength{\tabcolsep}{0in}
	$\begin{tabular}{cc}
	\multicolumn{1}{m{.5\columnwidth}}{\centering\includegraphics[width=.3\columnwidth]{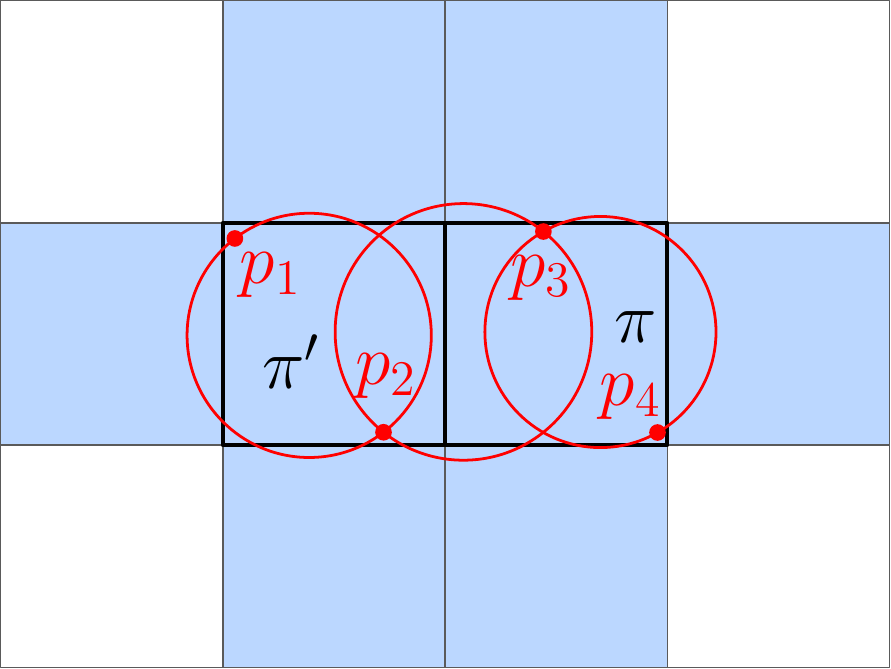}}
	&\multicolumn{1}{m{.5\columnwidth}}{\centering\includegraphics[width=.3\columnwidth]{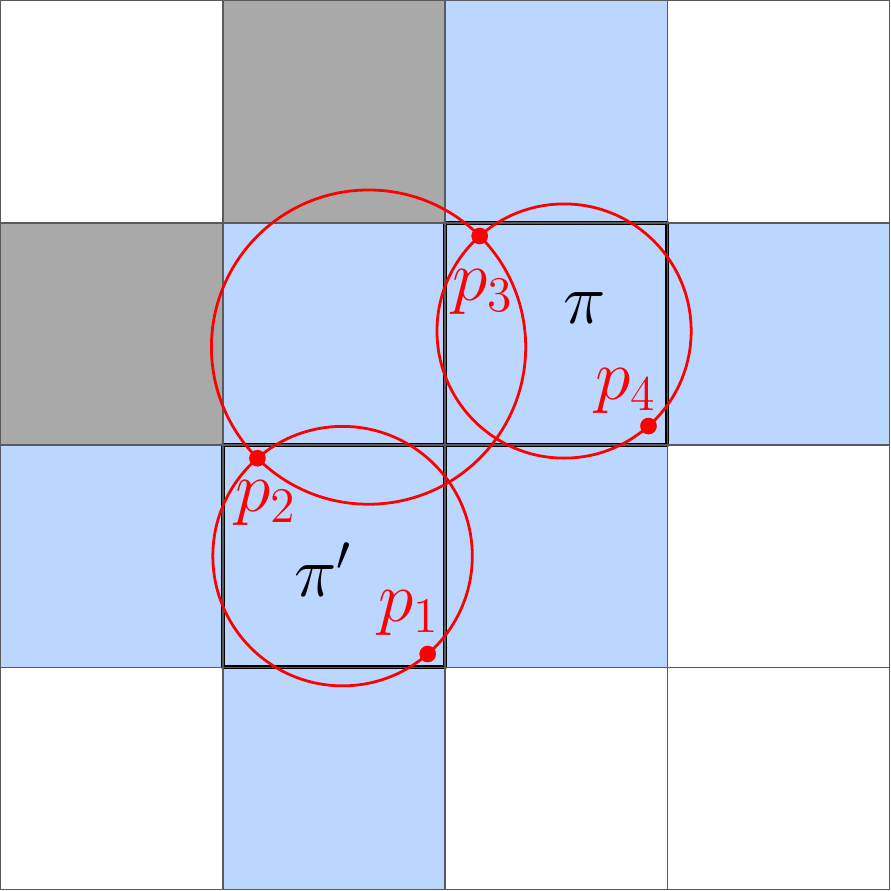}}\\
	(a)&(b)\\
	\multicolumn{1}{m{.5\columnwidth}}{\centering\includegraphics[width=.35\columnwidth]{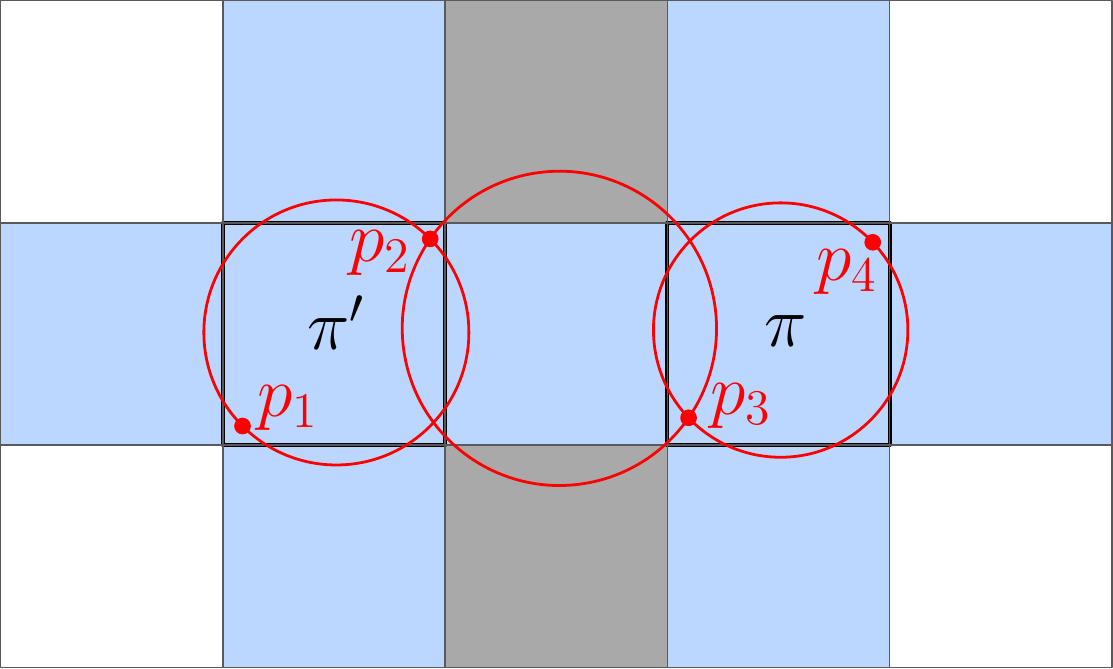}}
	&\multicolumn{1}{m{.5\columnwidth}}{\centering\includegraphics[width=.35\columnwidth]{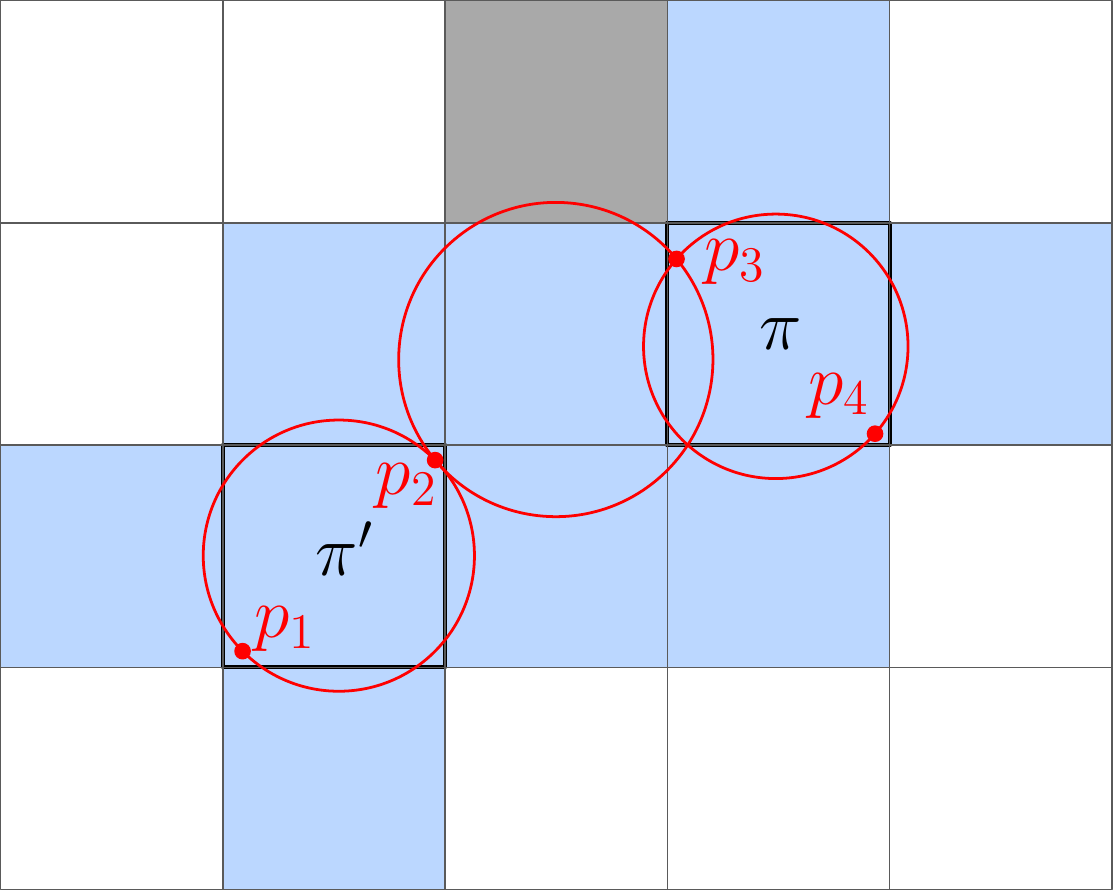}}\\
	(c)&(d)
	\end{tabular}$
	\caption{Illustration of the cells that are intersected by $D(p_1,p_2)\cup D(p_2,p_3)\cup D(p_3,p_4)$.}
	\label{neighbor-cell-fig}
\end{figure}

\begin{lemma}
	\label{neighbor-cell-lemma}
	Consider two cells $\pi$ and $\pi'$. Let $p_1$ and $p_2$ be any two points in $\pi'$, and let $p_3$ and $p_4$ be any two points in $\pi$. Let $\mathcal{D}$ be the union of three disks $D(p_1,p_2)$, $D(p_2,p_3)$, and $D(p_3,p_4)$.
	Then the following statements hold:
	\begin{enumerate}
		\item[1.] If $|p_2p_3|\leqslant 1/\sqrt{2}$ then $\mathcal{D}$ intersects at most $8$ cells.
		\item[2.] If $|p_2p_3|\leqslant 1$, and $\pi$ and $\pi'$ are $+$-neighbors, then $\mathcal{D}$ intersects at most $8$ cells.
		\item[3.] If $|p_2p_3|\leqslant 1$, and $\pi$ and $\pi'$ are $\times$-neighbors, then $\mathcal{D}$ intersects at most $10$ cells.
		\item[4.] If $|p_2p_3|\leqslant 1$, and $\pi$ and $\pi'$ are not neighbors, then $\mathcal{D}$ intersects at most $11$ cells.
	\end{enumerate}
\end{lemma}

\begin{proof}
We define $X$, as in Lemma~\ref{two-cell-lemma}, to be the set containing the cells $\pi$ and $\pi'$ and their $+$-neighbors (where $\pi$ and $\pi'$ play the roles of $\pi_p$ and $\pi_q$).
By Lemma~\ref{Circle-Square-lemma} the disk $D(p_1,p_2)$ can intersect only $\pi'$ and its four $+$-neighbors, and the disk $D(p_3,p_4)$ can intersect only $\pi$ and its four $+$-neighbors. Thus $D(p_1,p_2)$ and $D(p_3,p_4)$ can intersect only cells in $X$. Now we verify each statement. 

\begin{itemize}
	\item Statement 1. Since $|p_2p_3|\leqslant 1/\sqrt{2}$, the cells $\pi$ and $\pi'$ are neighbors, and thus their relative positions is among configurations A and B in Figures~\ref{config-fig}(a) and \ref{config-fig}(b). In each of these configurations the set $X$ contains $8$ cells. By Lemma~\ref{two-cell-lemma}, $D(p_2,p_3)$ does not intersect any cell outside $X$. Therefore, the union of the three disks, i.e. $\mathcal{D}$, intersects at most $8$ cells.  
	\item Statement 2. Since $\pi$ and $\pi'$ are $+$-neighbors, their relative position is configuration A. In this configuration, $X$ contains $8$ cells, and --- as we have seen in the proof of statement 2 of Lemma~\ref{two-cell-lemma} --- the disk $D(p_2,p_3)$ does not intersect any cell outside $X$. Therefore, $\mathcal{D}$ intersects at most $8$ cells; see Figure~\ref{neighbor-cell-fig}(a).  
	\item Statement 3. The relative position of $\pi$ and $\pi'$ is configuration B. In this configuration, $X$ contains $8$ cells, and --- as we have seen in the proof of statement 2 of Lemma~\ref{two-cell-lemma} --- the disk $D(p_2,p_3)$ intersects at most two cells outside $X$. Therefore, $\mathcal{D}$ intersects at most $10$ cells ($8$ cells in $X$ and 2 cells outside $X$); see Figure~\ref{neighbor-cell-fig}(b). 
	\item Statement 4. The relative positions of $\pi$ and $\pi'$ is among configurations C and D. In configuration C, $X$ contains $9$ cells, and by the proof of Lemma~\ref{two-cell-lemma} the disk $D(p_2,p_3)$ intersects at most two cells outside $X$. Therefore, $\mathcal{D}$ intersects at most $11$ cells; see Figure~\ref{neighbor-cell-fig}(c). In configuration D, $X$ contains $10$ cells, and by the proof of Lemma~\ref{two-cell-lemma} (statement 2, configuration D) the disk $D(p_2,p_3)$ intersects at most one cell outside $X$. Therefore, $\mathcal{D}$ intersects at most $11$ cells; see Figure~\ref{neighbor-cell-fig}(d).\qedhere 
\end{itemize}
\end{proof}

\subsubsection{Analysis of Hop Stretch Factor}
\label{analysis-section}
With lemmas in the previous section, we have all tools for proving the hop stretch factor of $H$. Recall that no point of $P$ lies on a grid line of $\Gamma$, and thus every point of $P$ is in the interior of some square of $\Gamma$. Consider any edge $(u,v)\in \UDG{P}$, and notice that $|uv|\leqslant 1$. In this section we prove the existence of a path, of length at most $341$, between $u$ and $v$ in $H$. 
Depending on whether $u$ or $v$ belong to $S$, we have three cases: (1) $u\notin S$ and $v\notin S$, (2) $u\in S$ and $v\in S$, and (3) $u\notin S$ and $v\in S$, or vice versa. These cases are treated using similar arguments. We give a detailed description of case (1) which gives rise to the worst stretch factor for our algorithm. We give a brief description of other cases at the end of this section. We denote by ``\pth{p}{q}'' a simple path between two points $p$ and $q$.

\vspace{8pt}
{\noindent\bf Case (1):} In this case $u,v\in P\setminus S$. Recall that, in $H$, $u$ and $v$ are connected to their closest visible vertices of $\DT{S}$; let $u_1$ and $v_1$ denote these vertices respectively. Therefore, in $H$, there is a \pth{u}{v} that consists of the edge $(u,u_1)$, a \pth{u_1}{v_1} in $\DT{S}$, and the edge $(v_1,v)$; see Figure~\ref{case3-fig}-top. In the following description we prove the existence of a \pth{u_1}{v_1} in $\DT{S}$ of desired length. By Lemma~\ref{unselected-length-lemma} we have $|uu_1|\leqslant 1/\sqrt{2}$ and $|vv_1|\leqslant 1/\sqrt{2}$; we will use these inequalities in our description. 

Let $\pi_u$, $\pi_v$, $\pi'_u$ and $\pi'_v$ denote the cells containing $u$, $v$, $u_1$ and $v_1$ respectively. Depending on the identicality of these cells we can have --- up to symmetry --- the following five sub-cases: (i) $\pi_u=\pi_v$ or (ii) $\pi'_u=\pi_u$ or (iii) $\pi'_u=\pi_v$ or (iv) $\pi'_u=\pi'_v$ or (v) all four cells are pairwise distinct.
These sub-cases are treated using similar arguments. We give a detailed description of sub-case (v) which gives rise to the worst stretch factor for our algorithm. We give a brief description of other sub-cases at the end of this section..  

Assume that $\pi_u$, $\pi_v$, $\pi'_u$ and $\pi'_v$ are pairwise distinct. Since $|uu_1|\leqslant 1/\sqrt{2}$, $\pi_u$ and $\pi'_u$ are neighbors; similarly $\pi_v$ and $\pi'_v$ are neighbors.
Since $(u,u_1)\in\UDG{P}$, by properties (P2) and (P3) in Lemma~\ref{S-lemma} there exist two points $u_2,u_3\in S$ such that $u_2\in \pi'_u$, $u_3\in \pi_u$, and $|u_2u_3|\leqslant 1$. Similarly, there exist two points $v_2,v_3\in S$ such that $v_2\in \pi'_v$, $v_3\in \pi_v$, and $|v_2v_3|\leqslant 1$. Moreover, since $(u,v)\in\UDG{P}$, there exist two points $u_4,v_4\in S$ such that $u_4\in \pi_u$, $v_4\in \pi_v$, and $|u_4v_4|\leqslant 1$. See Figure~\ref{case3-fig}. It might be the case that $u_1=u_2$, $u_3=u_4$, $v_3=v_4$, or $v_1=v_2$. Since $u_1$ and $u_2$ are in the same cell, $|u_1u_2|\leqslant 1$; similarly $|u_3u_4|\leqslant 1$, $|v_1v_2|\leqslant 1$, and $|v_3v_4|\leqslant 1$.
Having these distance constraints, Corollary~\ref{Delaunay-cor} implies that in $\DT{S}$ there exists a walk between $u_1$ and $v_1$ that consists of a \pth{u_1}{u_2} in $D(u_1,u_2)$, a \pth{u_2}{u_3} in $D(u_2,u_3)$, a \pth{u_3}{u_4} in $D(u_3,u_4)$, a \pth{u_4}{v_4} in $D(u_4,v_4)$, a \pth{v_4}{v_3} in $D(v_4,v_3)$, a \pth{v_3}{v_2} in $D(v_3,v_2)$, and a \pth{v_2}{v_1} in $D(v_2,v_1)$. Thus, there is a \pth{u_1}{v_1} in $\DT{S}$ that lies in the union of these seven disks; see Figure~\ref{case3-fig}. 

Let $\mathcal{D}$ denote the union of the seven disks. We want to obtain an upper bound on the number of cells intersected by $\mathcal{D}$. To that end, set $\mathcal{D}_u=D(u_1,u_2)\cup D(u_2,u_3)\cup D(u_3,u_4)$, and $\mathcal{D}_v=D(v_4,v_3)\cup D(v_3,v_2)\cup D(v_2,v_1)$. 
Define $X_u$ as the set containing the cells $\pi_u$ and $\pi'_u$ and their $+$-neighbors. Since $\pi_u$ and $\pi'_u$ are neighbors, their relative positions is among configurations A and B (Figures~\ref{config-fig}(a) and \ref{config-fig}(b)); in these configurations $X_u$ contains $8$ cells. Analogously, define $X_v$ with respect to $\pi_v$ and $\pi'_v$, and notice that $X_v$ also contains $8$ cells.

\begin{claim*}
	{\em Each of $\mathcal{D}_u$ and $\mathcal{D}_v$ intersects at most $8$ cells. Moreover, the cells that are intersected by $\mathcal{D}_u$ and $\mathcal{D}_v$ belong to $X_u$ and $X_v$, respectively.} 
\end{claim*}
\begin{proof}
	Because of symmetry, we prove this claim only for $\mathcal{D}_u$. Recall that $\pi_u$ and $\pi'_u$ are neighbors. If $\pi_u$ and $\pi'_u$ are $+$-neighbors, then $\mathcal{D}_u$ intersects at most $8$ cells by statement 2 in Lemma~\ref{neighbor-cell-lemma}. The proof of statement 2 also implies that these (at most $8$) cells belong to $X_u$. If $\pi_u$ and $\pi'_u$ are $\times$-neighbors, then by property (P3) in Lemma~\ref{S-lemma}, $(u_2,u_3)$ is the shortest edge of $\UDG{P}$ that runs between $\pi'_u$ and $\pi_u$. Since $(u_1,u)$ is also an edge between $\pi'_u$ and $\pi_u$, we have $|u_2u_3|\leqslant |u_1u|\leqslant 1/\sqrt{2}$. In this case $\mathcal{D}_u$ intersects at most $8$ cells by statement 1 in Lemma~\ref{neighbor-cell-lemma}. The proof of statement 1 implies that these cells belong to $X_u$.
\end{proof}

Notice that $\mathcal{D}=\mathcal{D}_u\cup \mathcal{D}_v\cup D(u_4,v_4)$. Based on this and the above claim, in order to obtain an upper bound on the number of cells that are intersected by $\mathcal{D}$ it suffices to obtain an upper bound on the number of cells, outside $X_u\cup X_v$, that are intersected by $D(u_4,v_4)$. To that end, define $X$ as the set containing the cells $\pi_u$ and $\pi_v$ and their $+$-neighbors, and notice that $X\subseteq X_u\cup X_v$. By Lemma~\ref{two-cell-lemma}, the disk $D(u_4,v_4)$ intersects at most $2$ cells outside $X$, and hence at most $2$ cells outside $X_u\cup X_v$. Therefore, the number of cells intersected by $\mathcal{D}$ is at most $|X_u\cup X_v|+2\leqslant 8+8+2=18$. Since by property (P1) in Lemma~\ref{S-lemma} each cell contains at most $20$ points of $S$, the set $\mathcal{D}$ contains at most $360$ points of $S$. Therefore, the \pth{u_1}{v_1} in $\DT{S}$ has at most $360$ vertices, and hence at most $359$ edges. Thus, the \pth{u}{v} in $H$ has at most $361$ edges (including $(u,u_1)$ and $(v_1,v)$). 

With a closer look at relative positions of $\pi_u$ and $\pi_v$ we show that $\mathcal{D}$ in fact intersects at most $17$ cells. This would imply that the \pth{u}{v} has at most $341$ edges as claimed. To that end we consider four configurations A, B, C, and D for $\pi_u$ and $\pi_v$, which we refer to them as sub-cases (v)-A, (v)-B, (v)-C, and (v)-D, respectively.

\begin{itemize}
	\item (v)-A. In this case $X_u$ and $X_v$ share $\pi_u$ and $\pi_v$ and thus $|X_u\cup X_v|\leqslant 14$. Moreover, by the proof of Lemma~\ref{two-cell-lemma} the disk $D(u_4,v_4)$ does not intersect any cell outside $X_u\cup X_v$; see Figure~\ref{config-fig}(a). Thus $\mathcal{D}$ intersects at most $14$ cells.
	\item (v)-B. In this case $X_u$ and $X_v$ share at least two cells (two $+$-neighbors of $\pi_u$ and $\pi_v$) and thus $|X_u\cup X_v|\leqslant 14$. Moreover, by the proof of Lemma~\ref{two-cell-lemma} the disk $D(u_4,v_4)$ intersects at most two cells outside $X_u\cup X_v$; see Figure~\ref{config-fig}(b). Thus $\mathcal{D}$ intersects at most $16$ cells (the shaded cells in Figure~\ref{case3-fig}-top). 
	\item (v)-C. In this case $X_u$ and $X_v$ share at least one cell (one $+$-neighbor of $\pi_u$ and $\pi_v$), and by the proof of Lemma~\ref{two-cell-lemma} the disk $D(u_4,v_4)$ intersects at most two cells outside $X_u\cup X_v$; see Figure~\ref{config-fig}(c). Thus $\mathcal{D}$ intersects at most $17$ cells (the shaded cells in Figure~\ref{case3-fig}-middle).  
	\item (v)-D. In this case $X_u$ and $X_v$ may not share any cell, but by the proof of Lemma~\ref{two-cell-lemma} the disk $D(u_4,v_4)$ intersects at most one cell outside $X_u\cup X_v$; see Figure~\ref{config-fig}(d). Thus $\mathcal{D}$ intersects at most $17$ cells (the shaded cells in Figure~\ref{case3-fig}-bottom).  
\end{itemize}

Even though $\mathcal{D}$ may intersect exactly $17$ cells, at the end of this section we show a possibly of decreasing the upper bound on the length of the \pth{u}{v} even further. This will require more case analysis which we avoid.

\begin{figure}
	\centering
	$\begin{tabular}{c}
	\multicolumn{1}{m{.9\columnwidth}}{\centering\includegraphics[width=.5\columnwidth]{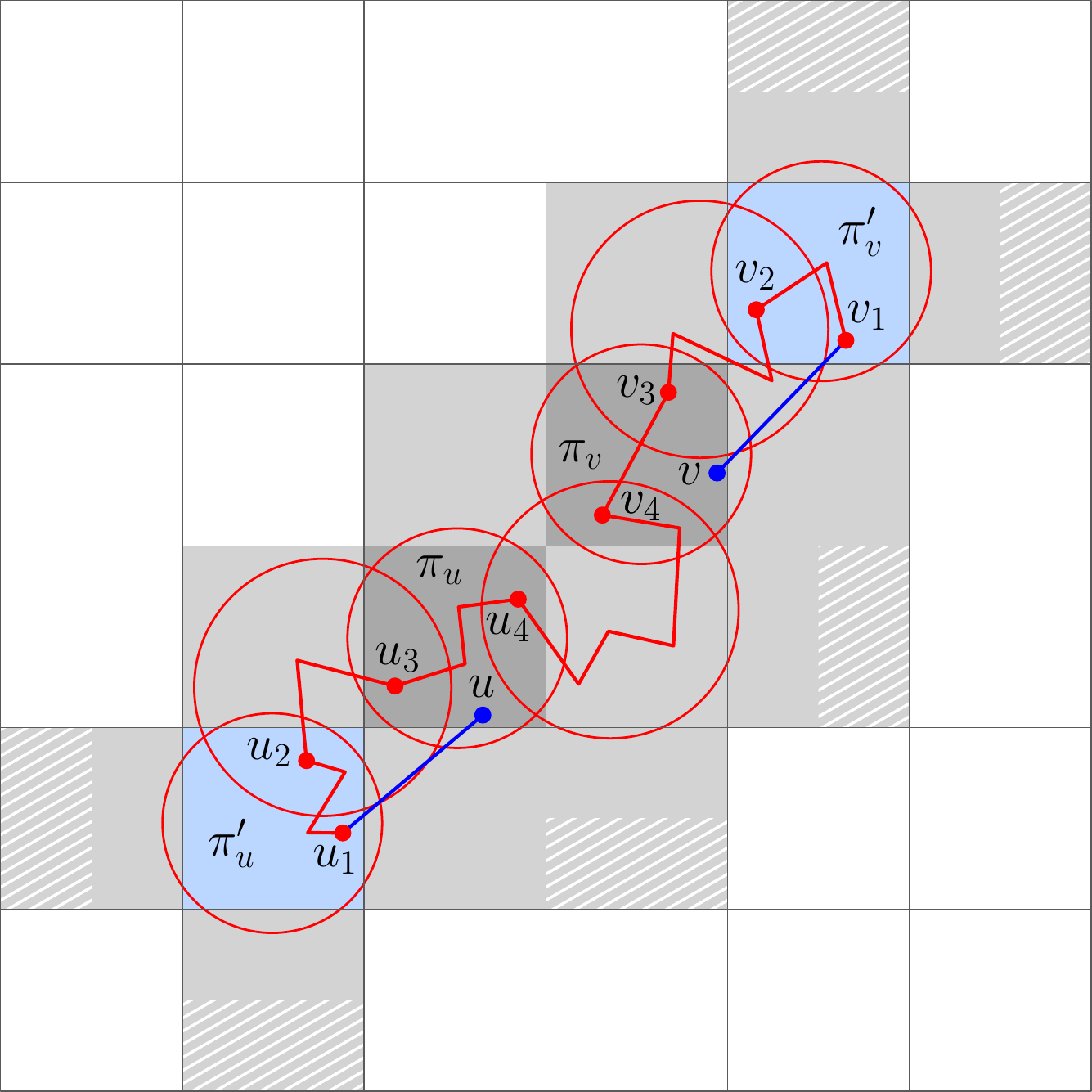}}\\
	\multicolumn{1}{m{.9\columnwidth}}{\centering\includegraphics[width=.6\columnwidth]{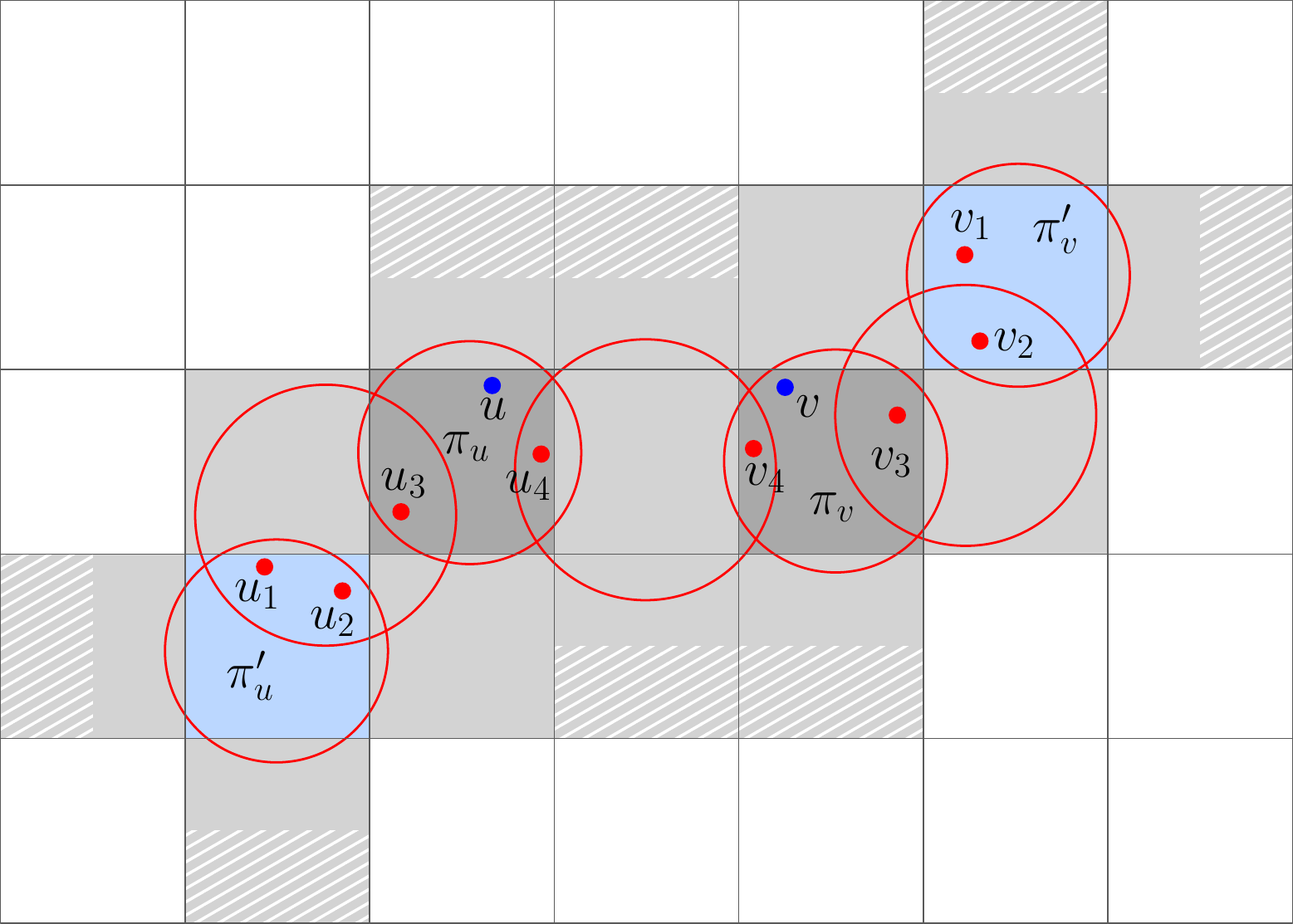}}\\ 
	\multicolumn{1}{m{.9\columnwidth}}{\centering\includegraphics[width=.55\columnwidth]{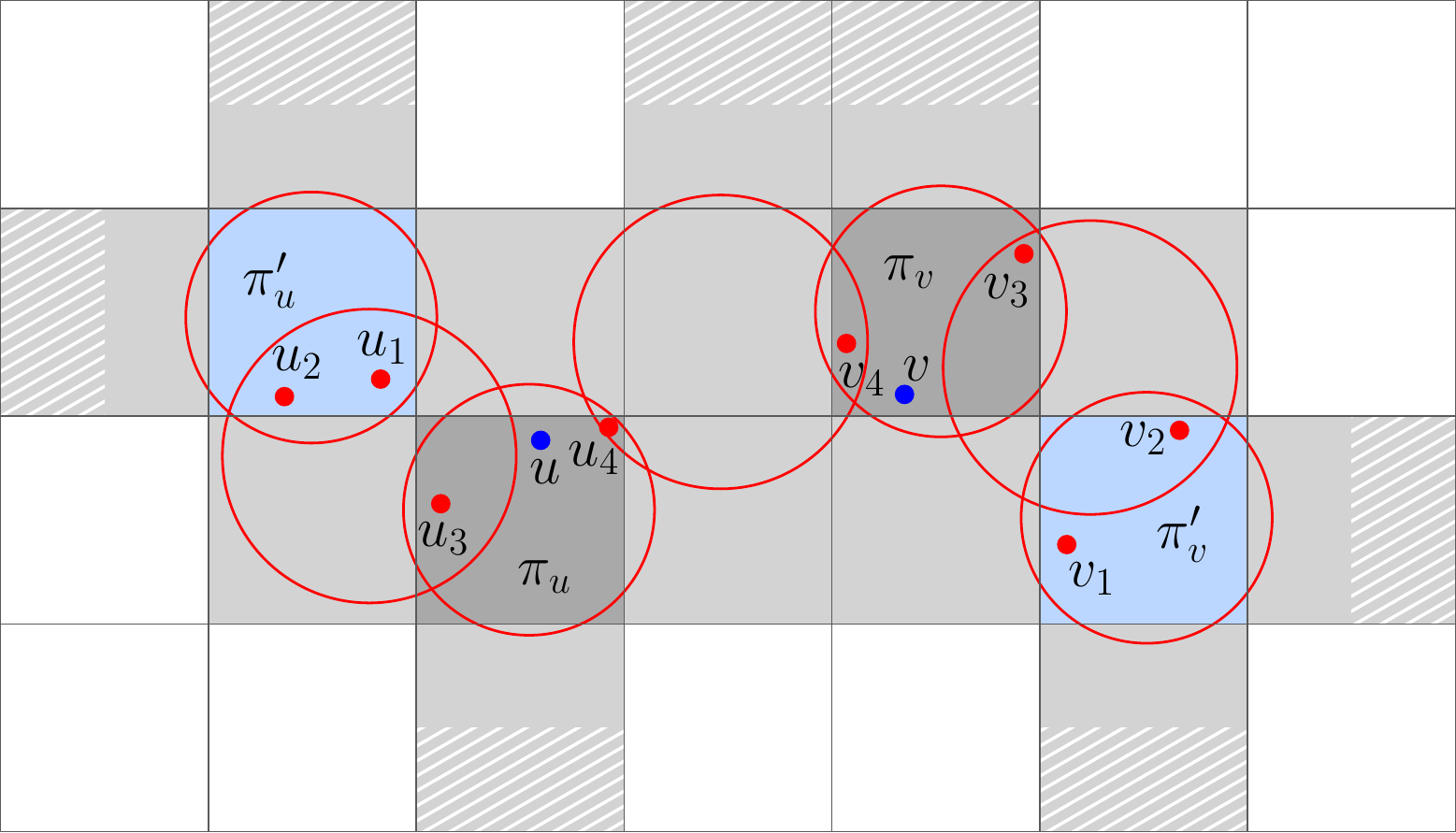}}
	\end{tabular}$
	\caption{Points $u, v,u_1,v_1$ belong to distinct cells. The cells $\pi_u$, $\pi_v$ are in configuration B (top), C (middle), and D (bottom). The red path in top figure corresponds to a \pth{u_1}{v_1} in $\DT{S}$.} 
	\label{case3-fig}
\end{figure}

\vspace{8pt}
{\noindent\bf Other Cases and Sub-Cases:}
We gave a detailed analysis for sub-case (v) (of case (1)) where $u$, $v$, $u_1$, $v_1$ lie in distinct cells $\pi_u$, $\pi_v$, $\pi'_u$, $\pi'_v$. Our analysis shows the existence of a \pth{u_1}{u_4} in $\mathcal{D}_u$ which intersects at most $8$ cells (which belong to $X_u$), and the existence of a \pth{v_1}{v_4} in $\mathcal{D}_v$ which  intersects at most $8$ cells (which belong to $X_v$). In the sequel we give short descriptions of case (2), case (3), and remaining sub-cases of case (1).    

Recall case (1) where $u,v\in P\setminus S$. In sub-case (i) where $\pi_u=\pi_v$, the sets $X_u$ and $X_v$ share at least $5$ cells ($\pi_u$ and its four $+$-neighbors). Therefore, $\mathcal{D}$ intersects at most $|X_u\cup X_v|+2\leqslant 8+8-5+2=13$ cells. Similarly, in each of sub-cases (iii) where $\pi'_u=\pi_v$ and (iv) where $\pi'_u=\pi'_v$, the sets $X_u$ and $X_v$ share at least $5$ cells, and thus $\mathcal{D}$ intersects at most $13$ cells.
In sub-case (ii) where $\pi'_u=\pi_u$, the set $X_u$ contains $5$ cells ($\pi_u$ and its four $+$-neighbors), and thus $\mathcal{D}$ intersects at most $5+8+2=15$ cells. Thus, in all these remaining sub-cases, the \pth{u}{v} has at most $301$ edges (including $(u, u_1)$ and $(v_1, v)$).

Now consider case (3) where $v\in S$ or $u\in S$ but not both. By symmetry we assume that $v\in S$, and thus $u\in P\setminus S$. In this case we do not have the point $v_1$ nor the cell $\pi'_v$; one may assume that $v_4=v_3=v_2$. Thus, $X_v$ contains at most $5$ cells ($\pi_v$ and its $+$-neighbors). By an argument similar to that of case (1), there exists a \pth{u_1}{v} in $\DT{S}$ that lies in $\mathcal{D}$ which intersects at most $|X_u\cup X_v|+2\leqslant 8+5+2=15$ cells. Therefore, there is a \pth{u}{v} in $H$ that has at most $300$ edges (including the edge $(u,u_1)$). 

Consider case (2) where $u,v\in S$. Since $|uv|\leqslant 1$, by Corollary~\ref{Delaunay-cor} there is a \pth{u}{v} in $\DT{S}$ that lies in $D(u,v)$. By Lemma~\ref{Circle-Square-lemma}, $D(u,v)$ intersects at most seven cells, and thus contains at most $140$ points of $S$. Therefore, the \pth{u}{v} has at most $139$ edges.

\vspace{8pt}
{\noindent\bf Further Improvement of Hop Stretch Factor:}
Recall the set $\mathcal{D}$, from Section~\ref{analysis-section}, which intersects at most $17$ cells, and hence contains at most $340$ points of $S$; this implied that the length of the \pth{u}{v} in $H$ is at most $341$. In this section we show a possibility of how one could improve the upper bound on the number of points in $\mathcal{D}$ to $319$; this would decrease the upper bound on the length of the \pth{u}{v} to $320$. However, to show this, one requires to go through some case analysis, which we avoid in this paper.

Recall that in sub-cases (v)-B, (v)-C, and (v)-D the set $\mathcal{D}$ intersects at most $16$, $17$, and $17$ cells respectively; see Figure~\ref{case3-fig} for an illustration of these sub-cases. In all other cases and sub-cases, $\mathcal{D}$ intersects at most $15$ cells, and hence contains at most $300$ points of $S$. Thus, it suffices to show, only for sub-cases (v)-B, (v)-C, (v)-D, that $\mathcal{D}$ contains at most $319$ points of $S$. This involves some case analysis which we provide an overview of that. 

For each cell $\pi$, let $\pi_L$ and $\pi_R$ be left and right rectangles obtained by bisecting $\pi$ with a vertical line, i.e., $\pi_L=\pi_{NW}\cup\pi_{SW}$ and $\pi_R=\pi_{NE}\cup \pi_{SE}$; see Figure~\ref{grid-fig}. Recall points $s_i$ in Lemma~\ref{S-lemma}, i.e., the points of $S$ that lie in $\pi$. All points $s_8,s_{12},s_{14}$ lie in $\pi_L$, and all points $s_{11},s_{13},s_{17}$ lie in $\pi_R$. Consider a similar partitioning of $\pi$ into top and bottom rectangles $\pi_T$ and $\pi_B$, and observe that all points $s_5,s_6,s_7$ lie in $\pi_T$ and all points $s_{18},s_{19},s_{20}$ lie in $\pi_B$. We refer to $\pi_L$, $\pi_R$, $\pi_T$, $\pi_B$ by ``rectangles''.
The tiled regions in Figure~\ref{case3-fig} correspond to these rectangles (from different cells). These rectangles (tiled regions in Figure~\ref{case3-fig}) are not intersected by $\mathcal{D}$. This and the fact that each rectangle counts for three points, imply that for each rectangle we could subtract $3$ from the number of points in $\mathcal{D}$. In Figures~\ref{case3-fig}-top, \ref{case3-fig}-middle, \ref{case3-fig}-bottom, which correspond to sub-cases (v)-B, (v)-C, (v)-D, the number of these rectangles is $6$, $8$, and $7$, respectively.

Consider any of sub-cases (v)-B, (v)-C and (v)-D. Fix the relative position of $\pi_u$ and $\pi_v$. By checking all possible configurations of neighboring cells $\pi'_u$ and $\pi'_v$ around $\pi_u$ and $\pi_v$, we always get either (i) at most $15$ cells that are intersected by $\mathcal{D}$, or (ii) exactly $16$ cells that are intersected by $\mathcal{D}$ and at least $1$ rectangle that is not intersected by $\mathcal{D}$, or (iii) exactly $17$ cells that are intersected by $\mathcal{D}$ and at least $7$ rectangle that are not intersected by $\mathcal{D}$. In case (i), the set $\mathcal{D}$ contains at most $300$ points of $S$. Since each rectangle counts for three points of $S$, in cases (ii) and (iii), the set $\mathcal{D}$ contains at most $320-3=317$ and $340-7*3=319$ points of $S$, respectively.

\bibliographystyle{abbrv}
\bibliography{Hop-Spanner}
\end{document}